\documentclass[11pt]{article}

\usepackage{times}
\usepackage{amsmath, amsthm, amssymb}
\usepackage{color}
\usepackage{graphicx}
\usepackage{multirow}
\usepackage{algorithm}
\usepackage[noend]{algorithmic}
\usepackage{cite}
\usepackage{epsfig}
\usepackage{url}
\usepackage{enumitem}

\newcommand{\Cost}{T}
\newcommand{\AlgL}{$\epsilon$-\textsc{Broadcast}}

\def\myval{-3pt}
\newcommand{\comment}[1]{}

\newtheorem{theorem}{Theorem}
\newtheorem{fact}{Fact}
\newtheorem{corollary}{Corollary}
\newtheorem{lemma}{Lemma}

\title{\fontsize{15.5}{13}\selectfont Making Evildoers Pay:\vspace{5pt}\\ Resource-Competitive Broadcast in Sensor Networks\vspace{-5pt}}
 
\author{Seth Gilbert \and Maxwell Young}

\begin{document}

\date{}

\maketitle

\begin{abstract} 
Consider a time-slotted, single-hop, wireless sensor network consisting of $n$ correct devices and and $f\cdot n$  Byzantine devices where $f\geq 0 $ is any constant; the Byzantine devices may or may not outnumber the correct ones. There exists a trusted sender Alice who wishes to deliver a message $m$ over a single channel to the correct devices. There is also an evil user Carol who controls the Byzantine devices and uses them to disrupt the communication channel. For a constant $k\geq 2$, the correct and Byzantine devices each possess a meager energy budget of $O(n^{1/k})$, Alice and Carol each possess a limited budget of $\tilde{O}(n^{1/k})$, and sending or listening in a slot incurs unit cost. This setup captures the inherent challenges of guaranteeing communication despite scarce resources and attacks on the network. Given this Alice versus Carol scenario, we ask: Is communication of $m$ feasible and, if so, at what cost? 

We develop a protocol which, for an arbitrarily small constant $\epsilon>0$, ensures that at least $(1-\epsilon)\,n$ correct devices receive $m$ with high probability. Furthermore, if Carol's devices expend $\Cost$ energy jamming the channel, then Alice and the correct devices each spend only $\tilde{O}(\Cost^{1/(k+1)})$. In other words, delaying the transmission of $m$ forces a jamming adversary to rapidly deplete its energy supply and, consequently, cease attacks on the network.\vspace{-0pt}
\end{abstract}

\section{ Introduction} 

Wireless sensors are continually shrinking, leading to increasingly dense networks built out of increasingly low-power devices.  The concept of dense wireless sensor networks (WSNs) was popularized by the Smart Dust project~\cite{warneke:smart} which provided the foundations for the well-known contemporary motes manufactured by Crossbow~\cite{crossbow_inc} and Dust Networks~\cite{dustnetworks}. While the size of  commercially available units is on the order of a few cubic centimeters, more recent endeavors such as SPECKNET~\cite{specknet} aim to reduce this to the cubic millimeter  scale~\cite{arvind:towards}. With the drive toward smaller wireless devices, it is not difficult to fathom the future deployment of highly dense WSNs and, indeed, the difficulties of communicating in such networks has been considered  previously by the research community~\cite{gau:reliable,chen:location,zhao:understanding}. 

In this paper, we address the challenge of communicating in a dense WSN given an adversary Carol who engages in malicious interference of the wireless medium. Such {\it jamming attacks} have received significant attention in recent years given the ease of perpetrating such attacks and their effectiveness (see~\cite{young:overcoming} and references therein). Jamming constitutes a form of denial-of-service attack that is particularly devastating given that WSN devices, including proposed future architectures~\cite{arvind:towards,chalasani:survey}, are severely energy constrained. Therefore, the prospects for achieving communication seem dire given an attacker who controls a large number of network devices and coordinates their combined resources to jam.

In this energy-starved setting, a sensible approach is to consider the rate at which a jamming adversary is required to expend energy relative to those devices attempting to overcome the jamming. If the adversary's total cost $\Cost$ is substantially higher, then preventing communication for any extended duration is prohibitively expensive and forces the adversary to quickly exhaust her energy supply. Here, a useful measure of cost is the number of slots during which a device is utilizing the channel. Specifically,  sending and listening operations dominate the operating costs of the Telos mote~\textemdash~$35$mW and $38$mW at $0$dBm, respectively~\textemdash~while sleeping incurs negligible cost on the order of $\mu W$. Such a {\it resource-competitive} approach~\cite{gilbert:resource} was first explicitly studied in~\cite{king:conflict,king:conflict_journal} where communication between two devices is guaranteed at an expected cost of $O(\Cost^{\varphi-1})=O(\Cost^{0.62})$ per device where $\varphi=\frac{1+\sqrt{5}}{2}$ is the golden ratio. This result places the adversary at a substantial disadvantage, but we ask: Is it possible to do better?

To address this question, we consider a general network scenario involving $(f+1)n$ devices where $f$ is any positive constant. A powerful adversary Carol controls $t = f\cdot n$  Byzantine devices which may deviate from the protocol arbitrarily; we emphasize that $f>1$ is allowed. Given this attack model, a trusted sender Alice attempts to propagate a message $m$ to the remaining correct devices. All devices, correct and Byzantine, have a severely constrained energy budget. We exploit the following two insights: If a small {\it tunable} constant fraction of devices are allowed to terminate without receiving $m$, and we seek guarantees only with high probability, then significant improvements are possible. Given this, we derive the following main result:\vspace{-0pt}

\begin{theorem}\label{theorem:competitive_costs_general}
Let $k\geq 2$. Assume Alice has an individual budget of $O(n^{1/k}\ln^{k}n)$ and aims to deliver a message to $n$ correct nodes. Assume Carol is an adaptive adversary with an individual budget of $\tilde{O}(n^{1/k})$ who controls $f\cdot n$ Byzantine nodes for any constant $f\geq 0$. Each node, correct and Byzantine, possesses a budget of $O(n^{1/k})$. Then, for $n$ sufficiently large, there is a protocol that guarantees the following properties with high probability: \vspace{-3pt}
\begin{itemize}

\item If Carol and her  Byzantine nodes jam for $\Cost$ slots, then Alice and the correct nodes each incur an individual cost of only $\tilde{O}(\Cost^{\frac{1}{k+1}} + 1)$ and $O(\Cost^{\frac{1}{k+1}} + 1)$, respectively.\vspace{-3pt}

\item At least $(1-\epsilon)\,n$ correct nodes receive the message for any arbitrarily small constant $\epsilon>0$ and Alice and all correct nodes terminate within $O(n^{1+(1/k)})$ slots.\vspace{-3pt}

\end{itemize}
\noindent where $\tilde{O}$ denotes the existence of polylogarithmic terms. Additionally, if $f < 1/24$, then these results hold when Carol is also a reactive adversary.

\vspace{-3pt}
\end{theorem}

This type of ``almost-everywhere'' communication plays an important role in several distributed computing problems (see~\cite{censor:partial,georgiou:on,dolev:gossiping} and references therein). In many cases, it is sufficient to guarantee a majority of the processes receive critical information. For example, Alice and others may be attempting to implement Paxos~\cite{lamport:paxos}, which relies on the notion of a majority quorum; therefore, $m$ must reach a majority of the nodes. For any $t\leq(1-\delta)n$, for a constant $\delta>0$, our protocol guarantees this property. In general, the ability to reach a $(1-\epsilon)$-fraction of the network is likely to be of importance in emerging WSNs.\vspace{0pt}

\subsection{Alice versus Carol~\textemdash~Our Network Model}\label{subsection:model}\vspace{0pt}

We assume a single hop WSN with $(f+1)n$ devices where $f \geq 0$ is a constant and where $n$ is large; that is, the network is dense. Devices use a time division multiple access (TDMA)-like medium access control (MAC) protocol to access a single communication channel; time is divided into discrete {\it slots}, but no global broadcast schedule is assumed. 

Nodes can detect whether a channel is in use via {\it clear channel assessment (CCA)}~\cite{ramachandran:clear}. This is a common feature; for instance, it is available on the CC2420 transceiver~\cite{cc2420} of the Telos mote, and several theoretical models feature collision detection (see~\cite{alistarh:securing, richa:jamming2,gilbert:malicious,awerbuch:jamming,koo2}). Jamming is indistinguishable from the case when two or more legitimate messages collide over the channel. Furthermore, jamming or a collision can only be detected on the receiving end of the wireless channel and, when this occurs, any received data is discarded. Finally, we assume that the absence of channel activity cannot be forged; in practice, such forging would be difficult~\cite{capkun:integrity}.\smallskip
 
\noindent{\bf Network Participants:} There are $(f+1)n$ correct nodes in the network of which $t\leq f\cdot n$ suffer a Byzantine fault and may deviate {\it arbitrarily} from any prescribed protocol. Each node is limited by a sublinear {\it budget} of at most $C\,n^{1/k}$ for any constant integer $k\geq 2$ and a sufficiently large constant $C>0$.  

Messages sent by Alice can be authenticated. For example, scalable dissemination of a {\it small} number of public keys is possible and we may assume that her public key (and, perhaps, only hers) is known to all receivers. Other authentication schemes can be assumed~\cite{chan:key}; this is a partially-authenticated Byzantine model since Alice is the only participant who can be authenticated. Therefore, attempts to tamper with $m$ or spoof Alice can be detected. However, correct nodes may be spoofed which allows Carol to repeatedly request retransmissions of $m$ from Alice. Our protocol must be resource competitive despite such a {\it spoofing attack} and we address this in Section~\ref{subsection:request}.

Finally, as a trusted sender, Alice is invested in delivering information to the network; consequently, we expect her to bear more of the communication costs; however, given the scarce energy resources of WSNs, we still enforce a fairly strict budget on Alice. Specifically, for $k=2$, her budget is at most $C\, n^{1/k}\ln n$, the equivalent of only $O(\ln n)$ nodes; for general $k\geq 3$, her budget is at most  $C\, n^{1/k}\ln^k n$ (see Section~\ref{subsection:general_k}). Note that, for the purposes of symmetry, we concede the same to Carol (see below).

We do not assume that jamming has a uniform impact on the correct nodes.  Any jamming by Carol or her Byzantine nodes can cause collisions and lost messages for some participants, while others receive the message correctly. More formally, a {\it $\ell$-uniform adversary}~(see~\cite{richa:jamming2}) is one who may partition the nodes into at most $\ell$ sets, each of which experience a different jamming schedule. We assume the worst-case, that Carol is an $n$-uniform adversary; therefore, she selects which nodes may detect jamming on an individual basis.  Given the utility of collision detection, this capability yields a powerful advantage to Carol while abstracting many of the challenges to reliable wireless communication including hidden terminals and fading effects. Carol also possesses full information on how nodes have behaved (in terms of sending/listening) in the past and uses this knowledge to inform future attacks; that is, she   {\it adaptive}. In this extended abstract, we assume that the actions
of a node in the current slot are unknown to Carol; however, with modifications discussed in Section~\ref{section:noise}, our results hold when Carol possesses this information (ie. she is {\it reactive}) for $f < 1/24$. Finally, for the purposes of symmetry, when $k=2$, we  treat Carol as an additional Byzantine node with an individual  budget of $C\,n^{1/k}\ln n$ to match that of Alice; for general $k\geq 3$, her budget is at most  $C\, n^{1/k}\ln^k n$ (see Section~\ref{subsection:general_k}). \smallskip  

\noindent{\bf Our Goal:}  Alice wishes to  deliver $m$ to as many correct nodes as possible while Carol, along with her  Byzantine nodes, aims to prevent communication. Alice, Carol, correct nodes, and Byzantine nodes incur a unit cost of $1$ for sending, listening, jamming, or altering messages. 

We define ``with high probability'' (w.h.p.) to mean with probability at least $1-n^{-c}$ for some constant $c>0$ we can tune. Our goal is to design a protocol that guarantees w.h.p.  delivery of  $m$ to as many nodes as possible while ensuring the following two properties.  First, since all participants are energy starved, the protocol should be {\it load balanced}; that is, Alice and each correct node should incur asymptotically equal costs (up to logarithmic factors). Second, the costs incurred by Alice and each correct node should be asymptotically less than the total cost incurred by Carol and her nodes; that is, our protocol should be resource competitive.  \vspace{2pt}

\noindent{\bf A Note on Resource Competitiveness:} We aim to show that, while Carol and her Byzantine nodes may deplete their collective budget in attacking the network, each individual correct node spends relatively little in order to achieve communication. We are focusing on an individual correct node's cost compared to the aggregate cost incurred by Carol and her Byzantine nodes; call this a {\it local perspective}. But why not consider a {\it global perspective} by using the aggregate cost of Alice and her correct nodes for comparison? 

There are several points in response. First, resource competitiveness from a local perspective is not a trivial task, especially given the strict resource constraints placed on nodes. Consider the naive approach where a correct node continually sends $m$ until the jamming stops; this yields very poor resource competitiveness since {\it each} node spends at least as much as the adversary. Indeed, many  algorithms for communication in WSNs suffer similarly. Second, guaranteeing resource competitiveness from a local perspective bounds the relative cost incurred by any single node; that is, the adversary cannot force any particular node to spend a disproportionate amount relative the adversary. In terms of maximizing a the lifetime of a network, it might be undesirable to achieve a global advantage if some nodes end up incurring substantially more relative cost; therefore, a guarantee from a global perspective is not necessarily stronger. 
Furthermore, note that, from a global-perspective,  we are indeed achieving a constant-factor advantage when $f>1$. For further discussion  on aspects of resource competitive analysis. we refer the reader to~\cite{gilbert:resource}. 

\subsection{Related Work}\vspace{0pt}

There are a large number of results addressing general problems involving jamming attacks (see~\cite{young:overcoming}). Closely related to our work are the results in~\cite{king:conflict} which provide the first resource-competitive communication protocol for two devices and also address a simple scenario with $n$ devices. There are several differences between our current work and~\cite{king:conflict}. The latter provides Las Vegas protocols with expected costs and Carol's budget is completely unknown.  For the $2$-node and $n$-node scenarios in~\cite{king:conflict}, the corresponding protocols are not load balanced since Alice spends roughly $D^{0.62}$ while each correct receiving node spends $D$. Finally,~\cite{king:conflict} can tolerate a reactive adversary only if external background communication traffic exists at no cost. In contrast, here we sacrifice  a small number of nodes and focus on dense WSNs; however, our improved costs are guaranteed with w.h.p., our protocol is load-balanced, and the correct nodes themselves bear the costs for thwarting a reactive adversary. Our attack model differs in that we assume that both the correct and Byzantine nodes have roughly the same power. Specifically, Carol's collective budget is at most a constant-factor larger than the aggregate budget of the correct nodes and it is polynomially larger than any single correct node.


Work by Ashraf~{\it et al.}~\cite{ashraf:bankrupting} investigates a similar line of reasoning employing multi-block payloads, so-called ``look-alike packets'' (which bears some resemblance to our strategy for dealing with reactive adversaries in Section~\ref{section:noise}), and randomized wakeup times for receivers to force the adversary into expending more energy in order to effectively jam. Their approach is interesting but differs in many ways from our own and analytical results are not provided.

There are a number of relevant analytical results on jamming. Gilbert~{\it et al.}~\cite{gilbert:malicious} derive deterministic upper and lower bounds on the duration for which communication can be disrupted between two WSN devices where silence cannot be forged. Pelc and Peleg~\cite{pelc:feasibility} examine a random jamming adversary. Koo~{\it et al.}~\cite{koo2} examine the problem of multi-hop broadcast in a grid topology in the presence of jamming when the adversary's budget is exactly known. Awerbuch~{\it et al.}~\cite{awerbuch:jamming} give a jamming-resistant MAC protocol in a single-hop network with an adaptive, rate-limited bursty jammer. Richa~{\it et al}~\cite{richa:jamming2} significantly extended this work to multi-hop networks and, later, to reactive bursty adversaries~\cite{richa:jamming3}.
In models where mutiple channels are available,  Dolev~{\it et al.}~\cite{dolev:gossiping} address a $(1-\epsilon)$ gossiping problem, Gilbert~{\it et al.}~\cite{gilbert:interference} derive bounds on the time required for information exchange given a reactive adversary, and Dolev~{\it et al.}~\cite{dolev:secure} address secure communication while tolerating a non-reactive adversary. 

In addition to pursuing a resource-competitive approach, our work differs from these related works in several ways. Our adversary is $n$-uniform; many previous results assume a $1$-uniform adversary. Furthermore, our adversary can be both adaptive and reactive, and she does not necessarily adhere to a particular jamming strategy (ie. bursty or random). Finally, our protocol does not rely on the availability of multiple channels; something that would likely not hold true given that Carol controls $\Theta(n)$ nodes and the number of channels is quite limited in practice.\vspace{0pt}

\section{Our Algorithm}\label{section:algorithm}\vspace{0pt}

In this section, we focus on the case where $k=2$.  In Section 3, we present the algorithm for general $k$. Our communication algorithm \AlgL~is presented in Figure~\ref{fig:pseudocode1}. Recall we desire an algorithm that: (1) is  load balanced and (2) is resource competitive. The constant $\epsilon>0$ is the upper limit on the fraction of nodes that may terminate the protocol without receiving $m$; we assume it is set prior to deployment. For our analysis, let $\epsilon'>0$ be an arbitrarily small constant (see Section~\ref{subsection:request}) that we set and we will ``renormalize'' by $\epsilon'$ to obtain $\epsilon$ in the statement of our main result. A node $u$ is said to be {\it informed} if $u$ ever receives $m$; otherwise, $u$ is said to {\it uninformed}.  A slot that is either jammed or contains at least one transmission is called {\it noisy}; otherwise, it is called {\it silent}. 

\AlgL~proceeds in rounds indexed by $i$ incrementing from $1$ until communication is achieved. The two parameters $a$ and $b$ will be determined throughout the course of our analysis. Each round consists of three phases:\vspace{-0pt}
\begin{itemize}
\item{\it Inform Phase:} Consists of $2^{(a + b)i}$ slots. Alice send $m$ with probability $\frac{2\ln n}{2^{bi}}$
in each slot. Each node which has not yet received $m$ listens to a slot with probability $\frac{2}{\epsilon' \, 2^{(a + \frac{b}{2})i}}$.\vspace{-7pt}

\item{\it Propagation Phase:} Consists of $2^{(a + b)i}$ slots. Each node $u$ that received $m$ in the preceding inform phase sends $m$ with probability $\frac{1}{n}$ and then terminates at the end of the phase. Each uninformed node listens in each slot with probability $4e(c+1)/2^{(a + (b/2))i}$ for a sufficiently large constant $c>0$.\vspace{-5pt}

\item{\it Request Phase:} Consists of $2^{(\frac{b}{2}+1)i}$ slots.  In each slot, each uninformed node $u$ sends $\texttt{nack}$ with probability $1/n$ and listens with probability $\frac{c+1}{(1-e^{-64\epsilon'})\, 2^{i}}$. If at most $5c\ln n$ noisy slots are heard ($p$ cannot hear its own transmissions), then $u$ terminates. Alice listens with probability $\frac{c\ln n}{(1-e^{-4\epsilon'})\, 2^{((b/2)+1)i}}$ in each slot and she terminates if the number of noisy slots heard is at most $5c\ln n$. \vspace{-0pt}
\end{itemize}

\noindent{\bf Discussion:} Our protocol is parameterized by the two constants $a$ and $b$ and these values dictate the costs to Alice and each node, respectively. In designing our protocol, we do not force values onto $a$ and $b$; rather, these values are derived to achieve both load balancing and resource competitiveness. However, there are some self-evident bounds that we make explicit. Note that, in round $i=\lg{n}$, Alice's maximum expected cost is $\tilde{O}(n^a)$ which implies that $a\leq 1/2$ given the allowed budget. Similarly, each node's cost is $O(n^{b/2})$ which implies that $b\leq 1$.
 
We assume that the constant $C$ used in the budgets for Alice, Carol, and the nodes is large enough to subsume the constants in our protocol; see the details in~Section~\ref{section:correctness}, Lemma~\ref{lemma:competitive_costs}.  Finally, we note that there are two advantages of choosing to send/listen in each  slot independently and uniformly at random. First, our analysis is primarily concerned with $i=\Omega(\log\log n)$; therefore, the expected costs for both Alice and each node are $\Omega(\log n)$ which means that these costs can be bounded to within a constant factor of their expectation via standard Chernoff bounds. Therefore, our protocol's costs are guaranteed with high probability. Second, information of how Alice and each correct node has behaved in the past conveys no information  about their actions in the current slot. Therefore, our protocol does not yield any advantage to an adaptive adversary.\vspace{0pt}

\begin{figure}[t]
\begin{center}
{\small
\fbox{\parbox[t]{6.3in}{\noindent{}\AlgL~for round $i$ when $k=2$\vspace{-4pt}

\begin{itemize}\renewcommand{\labelitemii}{$\circ$} \setlength{\itemindent}{-7pt}

\item{\it Inform Phase -} In each of  $2^{(a + b)i}$ slots:\vspace{-4pt}
\begin{itemize}[leftmargin=5pt]

\item Alice sends $m$ with probability $\frac{2\ln n}{2^{bi}}$. \vspace{-2pt}

\item Each uninformed node listens with probability $\frac{2}{\epsilon'\,2^{(a + \frac{b}{2})i}}$.  
\end{itemize}\vspace{-5pt}

\item{\it Propagation Phase -} In each of $2^{(a + b)i}$ slots:\vspace{-4pt}
\begin{itemize}[leftmargin=5pt] 
\item Each informed node sends $m$ with probability $\frac{1}{n}$  and terminates at the end of the phase.

\item Each uninformed node listens with probability $\frac{4e(c+1)}{2^{ai+(b/2)i}}$.
\end{itemize}\vspace{-5pt}

\item{\it Request Phase -} In each of $2^{(b/2+1)i}$ slots:\vspace{-4pt} 
\begin{itemize} [leftmargin=5pt]

\item Each uninformed node sends \texttt{nack} with probability $\frac{1}{n}$,  listens with probability $\frac{c+1}{(1-e^{-64\epsilon'})2^{i}}$, and terminates if at most $5c\ln n$ noisy slots are heard. 

\item Alice listens with probability $\frac{c\ln n}{(1-e^{-4\epsilon'})2^{(b/2 + 1)i}}$ and terminates if at most $5c\ln n$ \texttt{nack} messages or noisy slots are heard. 
 
\end{itemize}
\end{itemize}
}
} 

} 
\end{center}
\vspace{-18pt}
\caption{Pseudocode for round $i$ when $k=2$.}
\vspace{-5pt}\label{fig:pseudocode1}\end{figure}

\subsection{Analysis of our Protocol}\label{section:analysis}

For the inform phase, let $X_u=1$ if a node $u$ receives $m$, otherwise let $X_u=0$. Note that, for nodes two different nodes $u$ and $v$, $X_u$ and $X_v$ are dependent variables. For example, if $X_u=0$ because Alice never sent $m$ or she was blocked, then it is more likely that $X_v=0$. Similarly, if $X_u=1$, then it is more likely that $X_v=1$. The following concentration result from~\cite{dubhashi:concentration} is useful:\vspace{-1pt}

\begin{theorem}(\cite{dubhashi:concentration})\label{theorem:concentration}
Let $X_1, ... , X_\ell$ be random variables. Let $f$ be a function such that for each $i \in  \{1, . . . , \ell\}$ there is a $c_i\geq 0$ such that $|~E[~f~|~X_1 , . . . , X_i ] - E[~f~| X_1 , . . . , X_{i-1}]~| \leq c_i$. Then:\vspace{-5pt}

$$Pr(f \geq E[X] + \lambda) < e^{-\frac{\lambda^2}{2\sum_{i=1}^{\ell} c_i^2}}$$\vspace{-15pt} $$\vspace{-10pt}Pr(f \leq E[X] - \lambda) < e^{-\frac{\lambda^2}{2\sum_{i=1}^{\ell} c_i^2}}  $$\vspace{-5pt}

\end{theorem}

\noindent{}Theorem~\ref{theorem:concentration} applies to {\it dependent variables}. Using this result, we show that, if Carol does not perform too much jamming, then w.h.p. there exists a set containing at least $\Theta(\frac{n\ln n}{2^{(b/2)i}})$ informed nodes by the end of the inform phase. We define an inform phase as {\it blocked} if more than half of the slots in this phase are jammed; otherwise, the phase is {\it unblocked}. In a blocked inform phase, Carol decides which nodes, if any, receive $m$ since she is $n$-uniform. We also make use of the following identity:\vspace{-1pt}

\begin{fact}\label{lemma:identity2}
$1-y \geq e^{-2y}$ for any $y\leq 1/2$.\vspace{-1pt}
\end{fact}

Throughout our analysis, we are concerned with $3\lg\ln n \leq i \leq \lg n + O(1)$ as these allow us to derive concentration results; as we will see, the upper bound is a natural limit on the length of time our algorithm runs. When we speak of informed/uninformed nodes, this implicitly applies only to {\it correct} nodes.  

\begin{lemma}\label{lemma:inform1}
Assume at least $\epsilon'\,n$ nodes are uninformed and active at the start of an unblocked inform phase and $3\lg\ln n \leq i \leq \lg n + O(1)$. Then, w.h.p., the number of nodes that become newly informed by the end of this inform phase is at least $\frac{(1-\lambda)n\ln n}{2^{(b/2)i}}$ for some arbitrarily small constant $\lambda>0$ and for $n$ sufficiently large.
\end{lemma}
\begin{proof}
Let $s=2^{(a+b)i}$. Define a binary random variable such that $X_u=1$ if node $u$ obtains $m$ in the inform phase; otherwise, let $X_u=0$. Let $q_j=1$ if Carol does not jam in slot $j$ and let $q_j=0$ otherwise. Then $Pr(X_u=1) = 1 - Pr(u \mbox{~fails in inform phase})$ $=1 - \prod_{j=1}^{s}(1 -  Pr(u\mbox{~succeeds in slot~}j)) = $ $1-\prod_{j=1}^{s}(1-\frac{2\ln n}{2^{bi}}\frac{2}{\epsilon' 2^{(a+b/2)i}}\cdot q_j)$ $\geq 1 - e^{-\frac{4\ln n}{\epsilon' 2^{(a + (3/2)b)i}} \sum_{j=1}^s q_j} \geq 1 - e^{-\frac{2\ln n}{\epsilon' 2^{(b/2)i}}}$ given that $\sum_{j=1}^s q_j \geq s/2$ since the inform phase is not blocked. Let $y=\frac{\ln n}{\epsilon' 2^{(b/2)i}}$. By Fact~\ref{lemma:identity2}, it follows that $1 - y = 1 - \frac{\ln n}{\epsilon' 2^{(b/2)i}} \geq e^{-2\ln n/\epsilon' 2^{(b/2)i}}$ since $y\leq 1/2$ given the range of $i$. Therefore, we conclude that $Pr(X_u=1) \geq 1 - e^{-\frac{2\ln n}{\epsilon' 2^{(b/2)i}}} \geq \frac{\ln n}{\epsilon' 2^{(b/2)i}}$. Now let $f = \sum_{u=1}^{\delta n} X_u$ where $1 \geq \delta\geq\epsilon'$ and there are $\delta\,n \geq \epsilon'\,n$ uninformed nodes still active. By linearity of expectation, the expected number of nodes that receive $m$ in the inform phase is $E[f] \geq \frac{\delta n\ln n}{\epsilon'2^{(b/2)i}} \geq \frac{n\ln n}{2^{(b/2)i}}$. To prove a concentration result with dependent variables, we note that $|~E[~f~|~X_1, ..., X_u] - E[~f~|~X_1, ..., X_{u-1}~]~| \leq c_u = 1$ and use Theorem~\ref{theorem:concentration}.  For an arbitrarily small constant $\lambda>0$, it follows that $Pr(f < \frac{(1-\lambda)\,\delta\,n\ln n}{2^{(b/2)i}}) < e^{-\frac{\lambda^2\,n^2\ln^2n}{2^{bi}\,2n}}= e^{-\Theta(\lambda^2 \ln ^2 n)}$ since $i\leq \lg n + O(1)$. For sufficiently large $n$, this implies the desired upper bound result.
\end{proof}

\noindent Lemma~\ref{lemma:inform1} reveals the importance of Alice's $O(2^{ai}\ln n)$ budget as it facilitates a sufficiently large $S_1$. The upper bound is similar:\vspace{0pt}

\begin{lemma}\label{lemma:inform2}
Assume at least $\epsilon'\,n$ nodes are uninformed at the start of an unblocked inform phase and $3\lg\ln n \leq i \leq \lg n + O(1)$. Then, w.h.p., the number of nodes that become newly informed by the end of this inform phase is at most $\frac{(1+\lambda')4n\ln n}{\epsilon'\,2^{(b/2)i}}$ for an arbitrarily small  constant $\lambda'>0$ and for $n$ sufficiently large.\vspace{0pt}
\end{lemma}
\begin{proof}
Let $s=2^{(a+b)i}$ and defining $X_u$ the same way, we have $Pr(X_u=1) = 1-\prod_{j=1}^{s}(1-\frac{4\ln n}{\epsilon'\,2^{(a+(3/2)b)i}}\cdot q_j)$ and note that $\frac{4\ln n}{\epsilon'\,2^{(a+(3/2)b)i}}\leq 1/2$ for the range of $i$ and sufficiently large $n$. Therefore,  $Pr(X_u=1) = 1-\prod_{j=1}^{s}(1-\frac{4\ln n}{\epsilon'\,2^{(a+(3/2)b)i}}\cdot q_j) \leq 1 - e^{-\frac{4\ln n}{\epsilon' 2^{(b/2)i}}}$ using the fact that $\sum_{j=1}^s q_j \geq s/2$ and Fact~\ref{lemma:identity2}. Then $Pr(X_u=1) \leq  1 - e^{-\frac{4\ln n}{\epsilon' 2^{(b/2)i}}} \leq \frac{4\ln n}{\epsilon' 2^{(b/2)i}} $  where the  inequality follows from the standard $1-x \leq e^{-x}$. Therefore, the expected number of newly informed nodes is less than $\frac{4\,n\ln n}{\epsilon'\,2^{(b/2)i}}$. Using Theorem~\ref{theorem:concentration}, where $\lambda'>0$ is an arbitrarily small constant, the probability that we have more than $\frac{(1+\lambda')4n\ln n}{\epsilon\,2^{(b/2)i}}$ newly informed nodes is superpolynomially small in $n$. For $n$ sufficiently large, this yields the desired lower bound.\vspace{0pt}
\end{proof}

\noindent{}Therefore, so long as at least  $\epsilon'\,n$ nodes are uninformed and active, we can generate a set $S_i$ of at least $\Theta(\frac{n\ln n}{2^{(b/2)i}})$  newly informed nodes for $3\lg\ln n \leq i \leq \lg n + O(1)$; note, the size of this set is always sublinear in $n$. 
Moreover, the size of this set is decreasing as $i$ increases. This is due to the increasing length of the rounds and the limited energy afforded to each node. 

In round $i$ of the propagation phase, newly informed nodes in $S_i$ send $m$ to the remaining uninformed nodes. A propagation phase is {\it blocked} if more than half of the slots are jammed; otherwise, the phase is {\it unblocked}. Again, in a blocked propagation phase, Carol can decide which nodes receive $m$ since she is $n$-uniform.  

\begin{lemma}\label{lemma:inform_phase2}
Consider $(3/b)\lg\ln n \leq i \leq \lg n + O(1)$ and assume that the inform phase in round $i$ was not blocked. Then, if the propagation phase in round $i$ is not blocked, w.h.p. all nodes are informed by the end of the propagation phase. 
\end{lemma}
\begin{proof}
Let $s=2^{(a+b)i}$ be the number of slots and let $x$ be the number of newly informed nodes from the inform phase. Since the inform phase was not blocked,  Lemmas~\ref{lemma:inform1} and~\ref{lemma:inform2} guarantee w.h.p. that $\frac{(1-\lambda)n\ln n}{2^{(b/2)i}} \leq x \leq \frac{(1+\lambda)4n\ln n}{\epsilon'2^{(b/2)i}}$ for some arbitrarily small constant $\lambda>0$ . In a single slot, the probability that exactly one informed node in $S_i$ is sending is lower bounded by $x(\frac{1}{n})$ $(1-\frac{1}{n})^{x-1}$ $\geq \frac{(1-\lambda)n\ln n}{2^{(b/2)i}n}\cdot (1-\frac{1}{n})^{((1+\lambda)4n\ln n/(\epsilon'2^{(b/2)i})) -1}$ $\geq \frac{(1-\lambda)\ln n}{2^{(b/2)i}}\,$ $ e^{-8(1+\lambda)\ln n/(\epsilon'2^{(b/2)i})}$ $\geq \frac{(1-\lambda)\ln n}{e\cdot 2^{(b/2)i}} \geq  \frac{\ln n}{e\cdot 2^{(b/2)i+1}}$ where the second inequality follows by applying Fact~\ref{lemma:identity2}, the third follows by noting that $2^{(b/2)i} \in \omega(\ln n)$ for $i=(3/b)\lg\ln n$ (later we show $b=1$, thus keeping $i$ within proper range), and the fourth follows from setting $\lambda \leq 1/2$. Note that the sublinear upper bound on the size of $S_i$ prevents the probability of exactly one node sending from being too small. Therefore, the probability a particular uninformed node  does not receive $m$ in a single slot is at most $1- \frac{\ln n}{e\cdot 2^{(b/2)i+1}}\,\frac{4e(c+1)}{2^{ai+(b/2)i}}\,q_j$  where $q_j=0$ if Carol jams and $q_j=1$ if she does not. The probability of a specific active and uninformed node failing to obtain $m$  in this phase is at most $\prod_{j=1}^{s}(1-\frac{2(c+1)\ln n}{2^{(a+b)i}}\,q_j )$ $ \leq e^{-\frac{2(c+1)\ln n}{2^{(a+b)i}}\cdot \sum_1^{s} q_j}  < n^{-(c+1)}$ since $\sum_1^{s} q_j \geq \frac{s}{2}$. A union bound over all nodes yields the result.\vspace{0pt}
\end{proof}

Note that any communication from $S_i$ aimed at telling Alice the inform phase was successful could be spoofed by Carol. Therefore, $S_i$ cannot ever replace Alice (allowing her to sleep) since it is impossible to verify that $S_i$ was  created until the protocol terminates. Furthermore, keeping $S_i$ around for use in the propagation phase of round $i+1$ is wasteful since $S_{i+1}$ alone is sufficient. Increasing the sending probability of each node in $S_i$ is also wasteful and causes nodes to exceed their budget in later rounds.  Therefore, $S_i$ terminates at the end of every propagation phase.\vspace{0pt}

\subsection{Request Phase: Tolerating Spoofing}\label{subsection:request}\vspace{0pt}

A request phase in round $i$ is said to be blocked if Carol jams more than $(1-e^{-4\epsilon'})2^{(b/2 + 1)i}=\Omega(2^{(b/2+1)i})$ slots during the phase. Any constant fraction of the request phase will work; however, we choose this threshold to simplify the analysis. Again, recall that $3\lg\ln n \leq i \leq \lg n + O(1)$.  We state two important properties of Alice's termination condition:\vspace{-2pt} 

\begin{lemma}\label{lemma:Alice_request1}
Assume that the request phase is unblocked. If at most $2\epsilon'n$ nodes remain active, where then w.h.p. Alice (correctly) terminates the protocol for $\epsilon'\leq 1/2$.\vspace{0pt}
\end{lemma}
\begin{proof}
Let $s=2^{(b/2+1)i}$. The probability that no uninformed node sends a \texttt{nack} in a particular slot is $(1-\frac{1}{n})^{2\epsilon'n}$ for $\epsilon' \leq 1/2$. By Fact~\ref{lemma:identity2}, we have $(1-\frac{1}{n})^{2\epsilon'n}\geq e^{-4\epsilon'}$; therefore, the probability that a slot is noisy is $1 - (1-\frac{1}{n})^{2\epsilon'n} \leq 1- e^{-4\epsilon'}$. Let $Y_j=1$ if slot $j$ is noisy due to a \texttt{nack} message by an uninformed node; otherwise, let $Y_j=0$.  The expected number of noisy slots heard by Alice due to uninformed nodes is at most $E[\sum_{1}^{s} Y_j] \leq \sum_1^{s}\frac{c\ln n}{(1-e^{-4\epsilon'})2^{(b/2+1)i}}\cdot (1-e^{-4\epsilon'}) = c\ln n$. Pessimistically, assume that each of Carol's blocked slots occurs when none of the other uninformed nodes are sending a \texttt{nack} message. Let $Z_j=1$ if slot $j$ is noisy due to Carol jamming; otherwise, let $Z_j=0$. The expected number of jammed slots is at most $E[\sum_{1}^{s} Z_j] \leq \frac{c\ln n}{(1-e^{-4\epsilon'})2^{(b/2+1)i}}\cdot (1-e^{-4\epsilon'})2^{(b/2 + 1)i} = c\ln n$. Therefore, the total expected number of noisy slots that Alice hears is at most $2c\ln n$. By standard Chernoff bounds, the probability of exceeding $5c\ln n$ is at most $1/n^c$.\vspace{0pt}
\end{proof}

\begin{lemma}\label{lemma:Alice_request2}
Assume at least $32\epsilon'n$ nodes are active at the beginning of a request phase where $\epsilon' \leq 1/32$. Then. w.h.p., Alice (correctly) does not terminate.\vspace{0pt}
\end{lemma}
\begin{proof}
The bad event occurs if the number of noisy slots that Alice detects is less than $5c\ln n$. Since Carol cannot forge silence, we do not consider her behavior here. The probability of a noisy slot is $1-(1-\frac{1}{n})^{32\epsilon'n} \geq 1 - e^{-32\epsilon'}$. Therefore, as in the proof of Lemma~\ref{lemma:Alice_request1}, $E[Y] \geq \sum_1^{s}\frac{c\ln n}{(1-e^{-2\epsilon'})} \cdot (1-e^{-32\epsilon'})\geq 10c\ln n $ for  $s=2^{(b/2+1)i}$ and any $\epsilon' \leq 1/32$. The result follows by standard Chernoff bounds.\vspace{-5pt} 
\end{proof}

\noindent We prove a similar result for uninformed nodes, although the constants differ slightly and we discuss this below.\vspace{-5pt}

\begin{lemma}\label{lemma:node_request1}
Assume that  the request phase is unblocked. If at most $32\epsilon' n$ nodes are active, where $\epsilon' \leq 1/64$, then w.h.p. every node terminates by the end of the request phase.\vspace{0pt}
\end{lemma}
\comment{
\begin{proof}
The probability that no uninformed node sends a \texttt{nack} in a particular slot is $(1-\frac{1}{n})^{32\epsilon'n}$. We use the fact that $(1-\frac{1}{n})^{32\epsilon'n}\geq e^{-64\epsilon'}$; therefore, the probability that a slot is noisy is $1 - (1-\frac{1}{n})^{64\epsilon'n}  \leq 1- e^{-64\epsilon'}$. 
Let $Y_j=1$ if slot $j$ is noisy due to a \texttt{nack} message by an uninformed node; otherwise, let $Y_j=0$.  The expected number of noisy slots due to uninformed nodes is at most $E[\sum_{1}^{s} Y_j] \leq \sum_1^{s}\frac{c+1}{(1-e^{-64\epsilon'})2^{i}}\cdot (1-e^{-64\epsilon'}) = (c+1)\,2^{(b/2)i}$. Pessimistically, assume that each of Carol's blocked slots occurs when none of the other uninformed nodes are sending a \texttt{nack} message. Let $Z_j=1$ if slot $j$ is noisy due to Carol jamming; otherwise, let $Z_j=0$. The expected number of jammed slots is at most $E[\sum_{1}^{s} Z_j] \leq \frac{c+1}{(1-e^{-64\epsilon'})2^{i}}\cdot (1-e^{-4\epsilon'}) 2^{(b/2 + 1)i} = c\,2^{(b/2)i} \frac{1-e^{-4\epsilon'}}{1-e^{-64\epsilon'}} \leq (c+1)\,2^{(b/2)i}$ for $\epsilon'\leq 1/64$. Therefore, the total expected number of noisy slots heard by a node $u$ is at most $2(c+1)\ln n$ and, by standard Chernoff bounds and union bounding over all $n$ nodes, the probability that $u$ hears more than $5c\ln n$ is at most $1/n^c$.
\end{proof}
}
\begin{lemma}\label{lemma:node_request2}
Assume at least $1024\epsilon' n$ nodes are active at the beginning of a request phase where $\epsilon' \leq 1/1024$. The, w.h.p., none of the uninformed nodes terminate in that request phase.\vspace{0pt}
\end{lemma}
\comment{
\begin{proof}
The bad event occurs if the number of noisy slots that an uninformed node $u$ detects is less than $5c\ln n$. Since Carol cannot forge silence, we do not consider her behavior here. The probability of a noisy slot is at most $1-(1-\frac{1}{n})^{1024\epsilon'n} \geq 1 - e^{-1024\epsilon'}$. Let $Y_j=1$ if $u$ hears a noisy slot; otherwise, let $Y_j=0$. Note that $Pr(Y_j=1) \geq (1-\frac{1}{n})(\frac{c+1}{(1-e^{-64\epsilon})2^i})(1 - e^{-1024\epsilon'})$. Letting $Y=\sum_1^s Y_j$, we have $E[Y] \geq \sum_1^{s}(1-\frac{1}{n})\,\frac{c+1}{(1-e^{-64\epsilon'})2^i}\,(1-e^{-1024\epsilon'})\geq 10(c+1)\,2^{(b/2)i}$ for  $s=2^{(b/2+1)i}$ and any $\epsilon' \leq 1/1024$. Therefore, by standard Chernoff bounds and union bounding over all $n$ nodes, the probability that any uninformed node sees at most $5c\ln n$ is at most $1/n^{c}$. 
\end{proof}
 }
\noindent Critically, Alice should  only terminate  after the correct nodes terminate and, therefore,  our algorithm is designed in the following way.  If uninformed nodes are guaranteed w.h.p. to be active (a threshold of $1024\epsilon'n$ active nodes), then certainly Alice is guaranteed w.h.p. to be active (a threshold of $32\epsilon'n$ active nodes). Conversely, if Alice is guaranteed w.h.p. to have terminated  (a threshold of $2\epsilon'n$ active nodes), then the nodes are guaranteed  w.h.p. to have already terminated (a threshold of $32\epsilon'n$ active nodes).

To summarize the implications of these results, there are two bad situations: (1) if Alice or correct nodes can be tricked into perpetually executing the protocol at little cost to Carol, and (2) if Carol can cause Alice and all nodes to terminate with a large fraction of uninformed nodes. We have shown that (1) to keep Alice or nodes executing the protocol past their termination condition requires Carol to jam $\Omega(2^{(b/2+1)i})$ slots, and (2) w.h.p. Carol cannot force a large number of nodes to terminate without $m$.

\subsection{Correctness \& Resource Competitiveness}\label{section:correctness}\vspace{-1pt} 

The remainder of our analysis proceeds as follows. First, we show that if no blocked phases occur, then at least $(1-\epsilon')n$ nodes receive $m$. Second, when blocking phases do occur, we provide results on resource competitiveness. Finally, we prove that eventually a round is encountered where blocking must stop; consequently, at least $(1-\epsilon')n$ nodes become informed.\vspace{\myval}  

\begin{lemma}\label{lemma:no_blocking_cost}
Assume that there are no blocked phases in some round $i \geq 3\lg\ln n$ and at least $\epsilon' n$ nodes are active at the beginning of this round for any constant $\epsilon'>0$. Then, w.h.p., all correct nodes are informed and terminate.\vspace{\myval} 
\end{lemma}
\begin{proof}
Since the inform phase of round $i$ is not blocked, Lemmas~\ref{lemma:inform1} and~\ref{lemma:inform2} guarantee w.h.p. the creation of an $S_i$ of appropriate size. Since the propagation phase is not blocked, Lemma~\ref{lemma:inform_phase2} guarantees w.h.p. that all remaining active nodes receive $m$. Then, since the request phase is not blocked, Lemmas~\ref{lemma:Alice_request1} and~\ref{lemma:Alice_request2} guarantee w.h.p. that Alice terminates and Lemmas~\ref{lemma:node_request1} and~\ref{lemma:node_request2} guarantee that all nodes terminate.\vspace{\myval} 
\end{proof}

\noindent Lemma~\ref{lemma:no_blocking_cost} proves correctness in the absence of blocked phases; however, it is not yet apparent how the protocol may result in an small fraction of terminated, but uninformed, nodes. The critical observation is that we require $\epsilon' n$ active uninformed nodes at the beginning of the inform phase in order for Lemma~\ref{lemma:no_blocking_cost} to hold. 

Note that by blocking a propagation phase, an $n$-uniform  Carol may allow $2\epsilon'n$ nodes to remain uninformed and active. By Lemma~\ref{lemma:Alice_request1} and~\ref{lemma:node_request1}, Alice and all nodes then terminate.  Or Carol might block a propagate phase and let all but $32\epsilon' n$ nodes become informed; in this case, all nodes terminate with $32\epsilon' n$ uninformed. Critically, when Carol blocks an inform or propagate phase, she  decides how many nodes receive $m$ since she  is a $n$-uniform adversary; this illustrates the challenges posed by a $n$-uniform adversary.  We now analyze  resource competitiveness and begin by stating the costs when no blocked phases ever occur:\vspace{\myval} 

\begin{lemma}\label{lemma:no_block_costs}
Assume there are never any blocked phases. Then the cost to Alice is $O(\log^{3a+1}n)$ and the cost to each node is  $O(\log^{(3/2)b}n)$.\vspace{\myval} 
\end{lemma}
\begin{proof}
Given no blocked phases,  Lemma~\ref{lemma:no_blocking_cost} guarantees w.h.p. that all nodes become informed and terminate by round $i=3\lg\ln n$ and, given that round length increases geometrically with $i$, the costs in this round dominates that of the earlier rounds.  In the inform phase, Alice's cost is $O(\log^{3a+1}n)$ and each node's cost is $O(\log^{(3/2)b}n)$. In the propagation phase  Alice is inactive while each node in $S_1$ incurs a cost of $\tilde{O}(1)$, and each uninformed node incurs a cost of $O(\log^{(3/2)b}n)$.  In the request phase,  Alice's cost is $O(\log n)$ and each node's cost is $O(\log^{(3/2)b}n)$. Summing the costs yields the claim.\vspace{\myval} 
\end{proof}

\begin{lemma}\label{lemma:cost_ratios}
Assume that Carol spends $\Cost$ over the execution of~\AlgL~and at least one phase is blocked. Then, w.h.p the cost to Alice is $\max\{\tilde{O}(\Cost^{a/(a+b)}), \tilde{O}(\Cost^{a/(b/2 + 1)})\} $ and the cost to any node is $\max\{O(\Cost^{b/2(a+b)}), O(\Cost^{(b/2)/(b/2 + 1)})\}$.\vspace{-1pt} 
\end{lemma}
\begin{proof}
We analyze Alice and correct nodes separately:\smallskip 

\noindent{\it Cost for Alice:} There are two strategies by which Carol can prevent Alice from terminating. The first strategy is where Carol  blocks during at least one of the inform or propagation phases in each round. In this case, let $r$ be the first round where both the inform phase and propagation phase are not blocked. Then the cost to Carol is $\Cost=\Omega(2^{(a+b)r})$. Here, the cost to Alice is dominated by the cost of the next (and last) round since cost increases geometrically; this cost is $O(2^{ar}\ln n) =  O\left(\Cost^{a/(a+b)}\ln n\right)$. The second strategy occurs when  Carol blocks the request phase in order to trick Alice into believing that at least $32\epsilon' n$ nodes remain uninformed. Let $r'>r$ be the first round where Carol does not block the request phase. Note that $r'>r$ since it does Carol no good to block the request phase if the inform or propagate phases were already blocked in the round. Then, Carol's cost is $\Omega(2^{(b/2 + 1)r'})$ while the cost to Alice is $O(\Cost^{\frac{a}{b/2 + 1}}\ln n)$ since she will proceed into the next (and, w.h.p., final) round.

\smallskip 
 
\noindent{\it Cost for a Node:} There are two strategies by which Carol might prevent a node from terminating. The first strategy is where Carol  blocks at least one of the inform or propagation phases in each round. Let $r$ be the first round where this does not occur. Then, the cost to Carol is $\Cost=\Omega(2^{(a+b)r})$ while the cost to each node is $O(\Cost^{b/(2(a+b))})$. The second strategy occurs when  Carol blocks the request phase in order to trick the informed nodes into believing that at least $1024\epsilon' n$ nodes remain uninformed.  Let $r'>r$ be the first round where Carol does not block the request phase. Carol's cost is $\Omega(2^{(b/2+1)r'})$ while the cost to a node is $O(\Cost^{(b/2)/(b/2 + 1)})$ since the node will proceed into the next (and final) round.\vspace{\myval}  
\end{proof}

 
\noindent We now state our final result for $k=2$:

\begin{lemma}\label{lemma:competitive_costs}
Assume a sender Alice with a budget of at most $C\,n^{1/2}\ln n$ who aims to deliver a message $m$ to $n$ correct nodes. Assume an adaptive adversary Carol with an individual budget of $C\,n^{1/2}\ln n$ who controls an additional $n$  Byzantine nodes. Each correct and  Byzantine node possesses a budget of $C\,n^{1/2}$. Then, w.h.p., \AlgL~guarantees the following properties:\vspace{-2pt}
\begin{itemize}

\item If Carol and her  Byzantine nodes jam for $\Cost$ slots, then Alice and each correct node terminates with an individual cost of $\tilde{O}(\Cost^{\frac{1}{3}} + 1)$.\vspace{-3pt}

\item At least $(1-\epsilon)\,n$ correct correct nodes receive $m$ for an arbitrarily small  constant $\epsilon>0$ within $O(n^{3/2})$ slots.

\end{itemize}   
\end{lemma}
\begin{proof}
The worst-case resource-competitive ratios for Alice and the nodes should be equal. Lemma~\ref{lemma:cost_ratios} tells us that the exponents of interest for Alice are $a/(a+b)$ and $a/(b/2 + 1)$. Similarly, in the case of each node, the exponents of interest are $b/(2(a+b))$ and $(b/2)(b/2+1)$. Since we can choose $a\leq 1/2$ and $b\leq 1$, we can simplify Alice's maximum cost by setting $a/(a+b) = a/(b/2 + 1)$ and deriving the relationship $a + b(1-(1/2)) = 1$ (or, in general, $a + b(1-(1/k)) = 1$). Then we can relate this to a node's cost by setting $a/(a+b) = b/(2(a+b))$ and $a/(b/2+1) = (b/2)/(b/2+1)$ yields $a=b/2$ (or, for general $k$, $a = b/k$). Therefore, a common solution is $b=1$ and $a=1/2$ (or, in general, $a=1/k$). This yields a load-balanced solution where the cost to Alice is $O(\Cost^{1/3}\ln n + \ln^{5/2}n)$ and the cost to each node of $O(\Cost^{1/3} + \ln^{5/2}n)$ where the second cost term in each cost function follows from Lemma~\ref{lemma:no_block_costs}.

 
We now prove that the budgets of Alice and each node are sufficient to guarantee the claimed properties. When executing \AlgL, there exists some constant $d>0$ such that the  cost to each node in round $i$ is at most $d\, 2^{(b/2)i} = d\,2^{i/2}$; note that $d$ depends on the parameters $\epsilon$, $c$, and $k$ in our protocol. Recall that a blocked send, propagation, and request phase are defined slightly differently in terms of the constant fraction of slots jammed. Therefore, to simplify the analysis, redefine a blocked phase in round $i$ as one where more than $\beta\,2^{(3/2)i}$ slots are jammed for $0<\beta<1$; any positive constant in the range $(0,1)$ will yield the same resource competitive result asymptotically.

Each of the $t$ Byzantine nodes has a budget of $C\,n^{1/2}$; therefore, Carol and her Byzantine nodes possess a combined budget of $C\,f\,n^{3/2}$ $+ Cn^{1/2}\ln n \leq$ $C(f+1)n^{3/2}$. Therefore, using this budget, Carol cannot  block a send, propagation, or request  phase consisting of $(C/\beta)(f+1)\,n^{3/2}$ slots or more. Solving for $2^{(3/2)i} = (C/\beta)\,(f+1)\,n^{3/2}$ implies this occurs in round $i=\lg n + \frac{2}{3}\lg((C/\beta)(f+1))$. 

The cost to each correct node for executing \AlgL~in this round $i=\lg n + \frac{2}{3}\lg(2(C/\beta)(f+1))$ is at most $d\,2^{i/2} = d\,((C/\beta)(f+1))^{1/3}\,n^{1/2}$. We must take into account previous rounds, but given the doubling of length per round, the total cost up to this round is at most $2d\,((C/\beta)(f+1))^{1/3}\,n^{1/2}$. Therefore, so long as $C\geq (2d)^{3/2}\cdot ((f+1)/\beta)^{1/2}$, w.h.p. a correct node does not exceed its budget. By an almost identical argument, w.h.p., Alice does not exceed her budget of $C\,n^{1/2}\ln n$. Therefore, for $C$ sufficiently large, Alice and the correct nodes are guaranteed to reach a round where there are no blocked phases and, therefore, Lemma~\ref{lemma:no_blocking_cost} guarantees that at least $(1-\epsilon)n$ nodes are informed and terminate; the $\epsilon$-fraction that might terminate without $m$  follows from our observations about an $n$-uniform adversary. Given the doubling of the number of slots in each round,  this last phase occurs within $O(n^{3/2})$ slots. \vspace{\myval}  
\end{proof}

\noindent Clearly, no algorithm can disseminate a value from sender to {\it any} receivers in $o(n^{3/2})$ slots given that the total budget of Carol and her Byzantine nodes allow for the channel to be jammed continuously for that length of time. The following corollary is immediate:\vspace{\myval} 

\begin{corollary}
The latency for \AlgL~is asymptotically optimal.\vspace{\myval} 
\end{corollary}

\noindent Finally, we note that, in practice, each node may start with $i=1$ (or any other agreed upon constant) and run until at least its respective estimate of $d\lg\ln n$ is reached before terminating for some constant $d\geq 3$. That is, there is no need to start at exactly round $i=3\lg\ln n$; indeed, nodes may not agree on such a value and we discuss this further in Section~\ref{section:size}.


\section{The General Case}\label{subsection:general_k}\vspace{0pt}

For general $k$, it is not sufficient to simply replace $1/2$ by some function, say $(k-1)/k$, in our analysis since doing so results in a w.h.p. cost of $O(n^{(k-1)/k})$ rather than the desired $O(n^{1/k})$.  Instead, the propagation of $m$ must be extended in a non-trivial fashion by repeating the propagation phase $k-1$ times. For a fixed  round $i\in \Omega(\lg\ln n)$, we use these repeated propagation phases to prove the existence of sets of nodes $S_{i,h}$ where $h=1, ..., k-1$. The inform phase remains unchanged and results in the creation of the set $S_{i, 1}$ consisting of $\Theta(\frac{n\ln^{k-1} n}{2^{(1-1/k)i}})$ newly informed nodes. In turn, the propagation phase utilizes $S_{i,1}$ to guarantee the creation of $S_{i,2}$ which consists of  $\Theta(\frac{n\ln^{k-2} n}{2^{(1-2/k)i}})$ newly informed nodes. In general, throughout step $h$ of the propagation phase, the existing set $S_{i,h}$ of size   $\Theta(\frac{n\ln^{k-h} n}{2^{(1-h/k)i}})$  is used to create the new set $S_{i,h+1}$ of size  $\Theta(\frac{n\ln^{k-h-1} n}{2^{(1-(h+1)/k)i}})$ or larger. Therefore, by at least step $h=k-1$, the set $S_{i,k-1}$ contains   $\Theta(\frac{n\ln n}{2^{i/k}})$ informed nodes which  ensures that all remaining uninformed nodes can receive $m$ if no step in creating $S_{i,h}$ is blocked. Our pseudocode for general $k$ is given in Figure~\ref{fig:pseudocode_general_k} with the values for $a=1/k$ and $b=1$ substituted. \vspace{-2pt} 

\begin{figure}[t]
\begin{center}
{\small
\fbox{\parbox[t]{6.3in}{\noindent{}\AlgL~for round $i$\vspace{-4pt}

\begin{itemize}\renewcommand{\labelitemii}{$\circ$}\setlength{\itemindent}{-7pt}

\item{\it Inform Phase -} In each of  $2^{(1+ \frac{1}{k})i}$ slots:\vspace{-4pt}
\begin{itemize}[leftmargin=5pt]

\item Alice sends $m$ with probability $\frac{2c\ln^k n}{2^{i}}$. \vspace{-2pt}

\item Each uninformed node listens with probability $\frac{2}{\epsilon'\,2^{i}}$.  
\end{itemize}\vspace{-3pt}

\item{\it Propagation Phase -} For step $h=1$ to $ k-1 $ execute:\vspace{-6pt}

\begin{itemize}[leftmargin=5pt]\renewcommand{\labelitemiii}{-}

\item In each of  $2^{(1+\frac{1}{k})i}$ slots:\vspace{-0pt}
\begin{itemize}[leftmargin=10pt]
\item Each informed node sends $m$ with probability $\frac{1}{n}$ and terminates at the end of the step.

\item Each uninformed node listens with probability $\frac{2ec}{\epsilon'2^{i}}$.
\end{itemize}\vspace{-7pt}

\end{itemize}

\item{\it Request Phase -} In each of  $2^{(1+\frac{1}{k})i}$ slots:\vspace{-2pt} 
\begin{itemize}[leftmargin=5pt]

\item Each uninformed node sends a \texttt{nack} message with probability $\frac{1}{n}$,  listens with probability $\frac{c+1}{(1-e^{-64\epsilon'})2^{i}}$, and terminates if at most $5c\ln n$ noisy slots are heard. 

\item Alice listens with probability $\frac{c\ln n}{(1-e^{-4\epsilon'})2^{(1+ 1/k)i} }$ and terminates if at most $5c\ln n$ \texttt{nack} messages or noisy slots are heard. \vspace{-5pt}
 
\end{itemize}
\end{itemize}
}
} 

} 
\end{center}
\vspace{-18pt}
\caption{Pseudocode of round $i$ for  general $k$.}\label{fig:pseudocode_general_k}
\vspace{-9pt}\end{figure}

\subsection{Analysis for {\large $k=3$}}

We prove the case for $k=3$ case which demonstrates the key features regarding how the proof must change. Notable changes are that Theorem~\ref{theorem:concentration} no longer suffices to prove a lower bound on the size of the sets $S_{i,h}$; although, we can still use it to obtain a useful (loose) upper bound. Our proof structure changes to handle dependencies between the variables discussed in Section~\ref{section:analysis}. Also, Alice must now send with probability $2c\ln^k n/2^i$ and we consider $13\lg\ln n\leq i\leq \lg n + O(1)$; in general, $i = \Omega(\lg\ln n)$ for constant $k$. These changes are artifacts of our proof technique.

In the following, we conceptually divide the $n$ nodes into $\frac{n\ln^2 n}{2^{2i/3}}$  groups each of size $2^{2i/3}/\ln^2n$ nodes. We stress that these groups provide a method of counting how many nodes become informed, but such groupings play no part in the protocol. Our goal is to show that at least one member of each group becomes informed by the end of the phase.  To do this, we first prove that, for every group, each slot in the inform phase has a sufficient probability of being listened to by a node in the group.\vspace{\myval} 

\begin{lemma}\label{lemma:someone_listens_general2}
Assume at least $\epsilon'n$ nodes are uninformed and active. Then, for any slot in the inform phase, the probability that at least one node in each group is listening in that slot is at least $\frac{1}{2^{i/3}\ln^2n}$.\vspace{\myval} 
\end{lemma}
\begin{proof}
In round $i$, consider a group of $2^{2i/3}/\ln^2n$ nodes. Since at least $\epsilon'n$ nodes are uninformed and active, and the group membership is arbitrary, we can consider each such disjoint group to possess at least $\epsilon'\, 2^{2i/3}/\ln^2 n$ active nodes. Therefore, in a fixed slot, the probability that none of the nodes in a group are listening is $(1-\frac{2}{\epsilon'\, 2^{i}})^{\epsilon'\, 2^{2i/3}/\ln^2 n}\leq e^{-2/(2^{i/3}\ln^2 n)}$. Therefore, the probability that  at least one node in a group is listening is equal to or greater than $1-e^{-2/(2^{i/3}\ln^2 n)}\geq \frac{1}{2^{i/3}\ln^2n}$ by Fact~\ref{lemma:identity2}.\vspace{\myval} 
\end{proof}

\noindent{}Using our analysis via groups, the next lemma states proves the existence of a set $S_{i,1}$ of at least $\frac{ n\ln^2 n}{2^{2i/3}}$ nodes informed after the completion of a non-blocked inform phase.\vspace{-3pt} 

\begin{lemma}\label{lemma:inform_phase_general}
Assume that at least $\epsilon'n$ nodes are uninformed and active. Then, with high probability when $5\lg\ln n\leq i\leq \lg n + O(1)$, after an unblocked inform phase, there exist at least $\frac{n\ln^2n}{2^{2i/3}}$ newly informed nodes.\vspace{\myval} 
\end{lemma}
\begin{proof}
The phase consists of $s=2^{(4/3)i}$ slots. Since at least $\epsilon'n$ nodes are uninformed and active, Lemma~\ref{lemma:someone_listens_general2} guarantees that the probability no nodes in a group of size $2^{2i/3}/\ln^2 n$ receive $m$ in a fixed slot is at most $1-(\frac{2c\ln^3 n}{2^{i}})(\frac{1}{2^{i/3}\ln^2 n})q_j$  where $q_j=0$ if Carol jams and $q_j=1$ if she does not. It follows that, over the phase, the probability of all active uninformed nodes in the group failing  to obtain $m$  is at most $\prod_{j=1}^{s}(1-(\frac{2c\ln^3n}{2^{i}})(\frac{1}{2^{i/3}\ln^2n})q_j )$ $ \leq e^{-\frac{2c\ln n}{2^{(4/3)i}}\cdot \sum_1^{s} q_j}     \leq n^{-c}$ since $\sum_1^{s} q_j \geq \frac{s}{2}$. Taking a union bound over all groups, we conclude that at least one node from each of the $\frac{n\ln^2 n}{2^{2i/3}}$ groups becomes informed; this yields the result.\vspace{\myval} 
\end{proof}

\noindent{}While not tight, it is sufficient to use essentially the same upper bound argument as in Lemma~\ref{lemma:inform2}:\vspace{\myval} 

\begin{lemma}\label{lemma:inform_lower_general}
Assume at least $\epsilon'\,n$ nodes are uninformed and active at the start of an unblocked inform phase and $5\lg\ln n \leq i \leq \lg n + O(1)$. Then, w.h.p., the number of nodes that become newly informed by the end of this inform phase is at most $\frac{(1+\lambda')4n\ln^2n}{\epsilon'2^{i/2}}$ for an arbitrarily small constant $\lambda'>0$ and for $n$ sufficiently large.\vspace{\myval} 
\end{lemma}
\comment{
\begin{proof}
Let $s=2^{(4/3)i}$ and $X_u=1$ if node $u$ receives $m$ in the inform phase; otherwise, $X_u=0$. Then $Pr(X_u=1) = 1 - Pr( u \mbox{~fails in all $s$ slots} ) \leq 1 - \prod_{j=1}^s (1 - \frac{2c\ln^3 n}{2^i}\cdot \frac{2}{\epsilon'\,2^i}\,q_j) \leq 1 - e^{-(8c\ln^3 n)/(\epsilon'2^{2i}) \cdot \sum_{j=1}^s q_j}$ $\leq (8c\ln^3 n)/(\epsilon'2^{2i}) \cdot \sum_{j=1}^s q_j \leq 4c\ln^3 n/(\epsilon'2^{(2/3)i})$ since $\sum_{j=1}^s q_j \geq 2^{(4/3)i}/2$. Therefore, letting $X=\sum X_u$, we have $E[X] = O(n\ln^3 n/2^{(2/3)i})$; applying Theorem~\ref{theorem:concentration} directly as before does not work as the denominator of this expectation is too large. However, a loose upper bound is obtained by simply noting that $E[X] = O(n/2^{i/2})$ which suffices for our purposes (see Lemma~\ref{lemma:propagation_phase1_general}).  
\end{proof}
}
\noindent The members of the set $S_{i,1}$ are now used in the propagation phase to prove the existence of a larger set $S_{i,2}$ of size $\Theta(\frac{n\ln n}{2^{i/3}})$ informed members. Conceptually, this is proved by showing that members in $S_{i,1}$ can create at least one informed member in each of $\frac{n\ln n}{2^{i/3}}$ disjoint groups consisting of $\frac{2^{i/3}}{\ln n}$ nodes each; call each such conceptual group a {\it 2-group}. Again, we need to lower bound the probability that a slot is covered by at least one node in such 2-group.\vspace{\myval}

\begin{lemma}\label{lemma:someone_listens_general}
Assume that at least $\epsilon'n$ nodes are  uninformed and active, and assume that the inform phase was not blocked. Then, for any slot in  step $h=1$ of the propagation phase, the probability that at least one node in a 2-group is listening in that slot is at least $\frac{ec}{2^{(2/3)i}\ln n}$. \vspace{-6pt}
\end{lemma}
\begin{proof}
Let $G$ be a 2-group consisting of $2^{i/3}/\ln n$ nodes. Since at most $(1-\epsilon') n$ nodes have terminated, we can consider each disjoint 2-group to possess at least $\epsilon'\, 2^{i/3}/\ln n$ active nodes. The probability that none of the nodes in $G$ are listening in a slot is $(1-\frac{2ec}{\epsilon' 2^{i}})^{\epsilon'\cdot 2^{(i/3)}/\ln n}\leq e^{-2ec/(2^{(2/3)i}\ln n)}$. Therefore, the probability that  at least one node in $G$ is listening to a particular slot is at least $1-e^{-2ec/(2^{(2/3)i}\ln n)} \geq ec/(2^{(2/3)i}\ln n)$ by Fact~\ref{lemma:identity2}.\vspace{\myval}
\end{proof}

\noindent{}Define a blocked step of the propagation phase to be one where more than half the slots in that step are jammed. We can now prove the existence of $S_{i,2}$:\vspace{\myval}

\begin{lemma}\label{lemma:propagation_phase1_general}
Assume that at least $\epsilon'n$ nodes are uninformed and active, and the inform phase was not blocked. Then, w.h.p., at least $\frac{n\ln n}{2^{i/3}}$ nodes are newly informed in an unblocked step $h=1$ of the propagation phase for $5\lg\ln n\leq i\leq \lg n + O(1)$.\vspace{\myval}
\end{lemma}
\begin{proof}
The phase consists of $s=2^{(4/3)i}$ slots. For a fixed slot, the probability that a single node from $S_{i,1}$ is sending is $p=|S_{i,1}|(\frac{1}{n})(1-\frac{1}{n})^{|S_{i,1}|-1}$. By Lemmas~\ref{lemma:inform_phase_general} and~\ref{lemma:inform_lower_general}, we know that $\frac{n\ln^2n}{2^{2i/3}} \leq |S_{i,1}| \leq \frac{(1+\lambda')4n\ln^2n}{\epsilon'2^{i/2}}$, we have $p\geq \frac{\ln^2n}{2^{2i/3}}\cdot (1-\frac{1}{n})^{\frac{(1+\lambda')4n\ln^2n}{\epsilon'2^{i/2}}-1} \geq \frac{\ln^2 n}{2^{2i/3}} e^{-2(1+\lambda')4\ln^2n/(\epsilon'2^{i/2})} \geq \frac{\ln^2n}{e\,2^{2i/3}}$ in the range of $i$. Since at least $\epsilon'n$ nodes are uninformed and active, by Lemma~\ref{lemma:someone_listens_general}, the probability that no nodes in a fixed 2-group receive $m$ in a single slot is at most $1-(\frac{\ln^2 n}{e\,2^{2i/3}})(\frac{ec}{2^{(2/3)i}\ln n})q_j$  where $q_j=0$ if Carol jams and $q_j=1$ if she does not. It follows that the probability of all active and uninformed nodes in the $2$-group failing  to obtain $m$  in this step is at most $\prod_{j=1}^{s}(1-(\frac{c\ln n}{2^{(4/3)i}})q_j)$ $ \leq e^{-\frac{c\ln n}{2^{(4/3)i}}\cdot \sum_1^{s} q_j}  \leq n^{-c/2}$ since $\sum_1^{s} q_j \geq \frac{s}{2}$. Taking a union bound over all $n\ln n/2^{i/3}$ groups yields the result.\vspace{\myval}
\end{proof}

\noindent We need the next upper bound to ensure that members of $S_{i,2}$ will successfully send with sufficiently high probability:\vspace{\myval}


\begin{lemma}\label{lemma:inform_lower_general1}
Assume at least  $\epsilon'n$ nodes are active and uninformed, and both the inform phase and step $h=1$ of the propagation phase were not blocked. Then, w.h.p. where $5\lg\ln n \leq i \leq \lg n + O(1)$, the number of nodes that become newly informed by the end of this propagation phase is at most $\frac{(1+\lambda'')8cn\ln^2n}{\epsilon'2^{i/6}}$ for an arbitrary small constant $\lambda''>0$ and for $n$ sufficiently large.\vspace{\myval}
\end{lemma}

\comment{
\begin{proof}
Let $s=2^{(4/3)i}$ and $X_u=1$ if node $u$ receives $m$ in the inform phase; otherwise, $X_u=0$. The probability that a single node from $S_{i,1}$ is sending is $p=|S_{i,1}|(\frac{1}{n})(1-\frac{1}{n})^{|S_{i,1}|-1}$. By Lemmas~\ref{lemma:inform_phase_general} and~\ref{lemma:inform_lower_general}, we know that $\frac{n\ln^2n}{2^{2i/3}} \leq |S_{i,1}| \leq \frac{(1+\lambda')4n\ln^2n}{2^{i/2}}$, we have $p\leq   \frac{4(1+\lambda')\ln^2n}{2^{i/2}}\,(1 - \frac{1}{n})^{\frac{n\ln^2n}{2^{2i/3}}} \leq \frac{4(1+\lambda')\ln^2n}{e\,2^{i/2}}$ in the range of $i$ for $n$ sufficiently large. 

Then $Pr(X_u=1) = 1 - Pr( u \mbox{~fails in all $s$ slots} ) \leq 1 - \prod_{j=1}^s (1 - \frac{4(1+\lambda')\ln^2n}{e\,2^{i/2}}\cdot \frac{2}{\epsilon'\,2^i}\,q_j) \leq 1 - e^{-(16c\ln^2 n)/(\epsilon'2^{(3/2)i}) \cdot \sum_{j=1}^s q_j}$ $\leq 16c\ln^2 n/(e\epsilon'2^{(3/2)i}) \cdot \sum_{j=1}^s q_j \leq 8c\ln^2 n/(\epsilon'2^{(1/6)i})$. Therefore, letting $X=\sum X_u$, we have $E[X]  \leq 8c n\ln^2 n/(\epsilon'2^{(1/6)i})$. Applying Theorem~\ref{theorem:concentration}, the probability that we have more than $(1+\lambda'') \cdot \frac{8c n\ln^2 n}{2^{(1/6)i}}$ informed nodes is less than $e^{-\Theta\left(\frac{n^2\ln^4 n}{2^{(1/3)i}\,n}\right)}$ for a constant $\lambda''>0$ depending only on $\epsilon'$; this gives the result. 
\end{proof}
}

\noindent Finally, we can show that all remaining nodes receive $m$ if step $h=2$ of the propagation phase is not blocked:\vspace{\myval} 

\begin{lemma}\label{lemma:inform_phase2_general}
Let $7\lg\ln n \leq i \leq \lg n + O(1)$. Assume that in round $i$ both the inform phase and step $h=1$ of the propagation phase in round were not blocked. Then, if step $h=2$ of the propagation phase in round $i$ is unblocked, w.h.p. all nodes are informed by the end of the propagation phase. \vspace{0pt} 
\end{lemma}
\begin{proof}
Let $s=2^{(4/3)i}$ be the number of slots. Since the inform phase and step $h=1$ of the propagation phase were not blocked,  Lemmas~\ref{lemma:propagation_phase1_general} and~\ref{lemma:inform_lower_general1} guarantee w.h.p. that $\frac{n\ln n}{2^{i/3}} \leq |S_{i,2}| \leq \frac{(1+\lambda'')8cn\ln^2n}{\epsilon'2^{i/6}}$ for some arbitrarily small constant $\lambda>0$ . In a single slot, the probability that exactly one informed node in $S_{i,2}$ is sending is at least $|S_{i,2}|(\frac{1}{n})(1-\frac{1}{n})^{|S_{i,2}|-1}$ $\geq \frac{n\ln n}{2^{i/3}\,n}(1-\frac{1}{n})^{(1+\lambda'')8cn\ln^2 n/(\epsilon'2^{i/6})}$ $\geq \frac{\ln n}{e\,2^{i/3}}$ for $i\geq 13\lg\ln n $ and $n$ sufficiently large. Therefore, the probability a particular uninformed node  does not receive $m$ in a single slot is at most $1- \frac{\ln n}{e\,2^{i/3}} \frac{2ec}{2^{i}}\,q_j$  where $q_j=0$ if Carol jams and $q_j=0$ if she does not. The probability of a specific active and uninformed node failing to obtain $m$  in this phase is at most $\prod_{j=1}^{s}(1-\frac{2c\ln n}{2^{(4/3)i}}\,q_j )$ $ \leq e^{-\frac{2c\ln n}{2^{(4/3)i}}\cdot \sum_1^{s} q_j} $ which gives us the high probability guarantee since $\sum_1^{s} q_j \geq \frac{s}{2}$. Taking a union bound over all nodes yields the result.\vspace{\myval}
\end{proof}

\noindent{\bf Discussion:} Therefore,  all nodes will receive $m$ so long as neither the inform phase nor any  steps in the propagation phase are blocked. Note that the arguments in Lemma~\ref{lemma:cost_ratios} and~\ref{lemma:competitive_costs} do not change so long as $k$ is a constant (see below). Analogous to our argument when $k=2$, when round $i = \lg n + \frac{k}{k+1}\lg(C/\beta)$ is reached, Carol and the Byzantine nodes do not have sufficient energy to block a phase (or a step of a phase) in which case, the termination conditions for Alice and the correct nodes are met. By chaining together more proofs showing the existence of $S_{i,h}$, this proof structure can be extended for any constant $k$. \vspace{\myval} 

\subsection{Limits To This Approach}\label{subsection:limits}

By increasing $k$, the protocol is more resource competitive; however, there is a limit. Note that, due to the steps of the propagation phase, the latency and the overall costs increases by a factor of $\Theta(k)$. Now consider if $k\geq \omega(1)$. Then, Alice and her nodes each require $\omega(n^{1/k})$ to execute the $O(k)$  propagation phase steps and this exceeds their budget.

This cannot be remedied through any $\omega(1)$-factor increase in the budget of each node. To see why, let $k=2$ and assume that each node now has a budget of $C\,n^{1/2}\ln n$. Note that Carol may now block phases of length $C\,n^{3/2}\ln n$ which occurs for round $i=\lg n + (2/3)\lg\ln n + (2/3)\lg C$. However, in this round, each correct node must spend $2^{i/2}\ln n = \Omega(n^{1/2}\ln^{4/3}n)$; this exceeds its budget. This problem manifests for any $k=\omega(1)$.\vspace{\myval}  

\section{Extensions to the Protocol} 

In this section, we sketch how \AlgL~can be modified to tolerate a reactive adversary when $f <1/24$. We conclude by discussing how exact knowledge of $\ln n$ and $n$ is not required to successfully execute \AlgL. \vspace{0pt}

\subsection{Reactive Jamming: Make Your Own Noise}\label{section:noise}\vspace{0pt}

Within the current time slot, a reactive adversary can detect channel activity and decide whether to jam. The ability to perform CCA makes it possible for Carol to detect such activity based on the received signal strength indicator (RSSI) which incurs negligible cost. During either the inform or propagation phases, Carol is guaranteed to interfere with the transmission of $m$ if she jams. Such targeted jamming invalidates our analysis in Lemmas~\ref{lemma:inform1},~\ref{lemma:inform2}, and~\ref{lemma:inform_phase2}. However, while RSSI enables Carol to detect channel activity, it  provides no information about the transmitted content. Therefore, reactive jamming is only effective if the bulk of the channel activity involves the transmission of $m$. For example, if half of the slots contain non-critical traffic and the other half contain $m$, then jamming based simply on RSSI is no better than randomly jamming. While Carol might activate her transceiver in order to hear part of the transmission before deciding to jam, this is expensive. 


As in~\cite{king:conflict}, if there is sufficient background network traffic such that a constant fraction of the slots in each round are in use, then  a reactive adversary can be tolerated. But what if such traffic is absent? Another approach is to have the correct nodes generate their own traffic. Under this strategy, we show that, for $f < 1/24$, a reactive Carol is unable to prevent communication indefinitely and our algorithm is still resource competitive. Although $f< 1/24$ implies that the aggregate energy possessed by Alice and the correct nodes exceeds that of Carol and her Byzantine nodes, we emphasize that this problem is still non-trivial. For example, it is not possible to have have Alice outspend Carol since Alice can only send the message $O(n^{1/2})$ times while Carol can jam for $\Omega(n^{3/2})$ slots. As described above, we need to have the uninformed nodes generate additional traffic in order to overcome a reactive jammer.

To do this, for the inform and propagation phases, the modified protocol specifies that each node sends a decoy message  with probability $\frac{3}{4\epsilon'n}$ per slot and we assume that each correct node listens with a constant factor increase in probability (see $p_u$ in the proof of Lemma~\ref{lemma:inform1}). We now re-prove Lemma~\ref{lemma:inform1}.\vspace{5pt}

\noindent{\bf Lemma~\ref{lemma:inform1}.}{\it ~Assume at least $\epsilon'\,n$  nodes are uninformed and active at the start of an unblocked inform phase and $3\lg\ln n \leq i \leq \lg n + O(1)$. Then, w.h.p., the number of correct nodes that become newly informed by the end of this inform phase is at least $\frac{(1-\lambda)n\ln n}{2^{i/2}}$ for an arbitrarily small constant $\lambda>0$ and for $n$ sufficiently large.}\vspace{5pt}
\begin{proof}
Let $s=2^{(3/2)i}$. Let the random variable $Z_j=1$ if a slot is occupied by one or more decoy messages; otherwise, let $Z_j=0$. Since all correct nodes send a decoy message independently with uniform probability $\frac{3}{4\epsilon'n}$, $Pr(Z_j=1) = 1 - (1-\frac{3}{4\epsilon'n})^{\epsilon' n} \geq 1-e^{-3/4} \geq 1/2$. Letting $Z=\sum_j^s Z_j$ it follows that $E[Z]\geq (1/2)s$. Conversely, using Fact~\ref{lemma:identity2}, $Pr(Z_j=1) \leq 1 - (1-\frac{3}{4\epsilon' n})^n \leq 1 - e^{-3/(2\epsilon')}$; therefore, $E[Z] \leq (1 - e^{-3/(2\epsilon')})\,s$. By standard Chernoff bounds, for $i\geq 3\lg\ln n$, the number of slots containing one or more decoy messages, denoted by $s_N$, is $(1-\delta)(1/2)s \leq s_{N} \leq(1+\delta)(1 - e^{-3/(2\epsilon')})s$ w.h.p. for $\delta>0$ arbitrarily small depending only on sufficiently large $n$. By a similar argument, given $i\geq 3\lg\ln n$, the number of slots in which Alice sends $m$ is $s_{A} = (1\pm\delta')2^{i/2}\ln n$ w.h.p. for $\delta'>0$ arbitrarily small depending only on sufficiently large $n$.  

Redefine an inform phase as blocked if Carol jams more than $s/4$   slots {\it containing $m$ or at least one decoy message}. Carol's choice to jam such a slot (or listen to it) is made without knowing whether the slot contains $m$ or a decoy message. Therefore, for a fixed slot containing $m$ sent by Alice, the probability that Carol fails to listen to or block this slot in a non-blocked phase is at least $1 - \frac{s/4}{s_N}$. The probability that this same slot is not used by a correct node for sending a decoy message is at least $e^{-3/(2\epsilon')}$ as determined above. Let $p_u$ denote the probability that a node $u$ listens to a particular slot. Assuming Alice sends $m$, the probability that $u$ receives $m$ in a fixed slot is at least $(1 - \frac{s/4}{s_N})(e^{-3/(2\epsilon')})p_u \geq (e^{-3/(2\epsilon')} - \frac{2^{(3/2)i}}{4\cdot e^{3/(2\epsilon')}\cdot (1-\delta)(1/2)2^{(3/2)i}})p_u$  $= (\frac{1}{e^{3/(2\epsilon')}} - \frac{(1+\delta'')}{2e^{3/(2\epsilon')}})p_u$  for small enough $\delta$ given sufficiently large $n$. It follows that, for sufficiently small $\delta''$(say $\delta'' \leq 1/2$), the probability that $u$ receives $m$ in a slot is at least $(\frac{1}{e^{3/(2\epsilon')}} - \frac{(1+\delta'')}{2e^{3/(2\epsilon')}})p_u = \frac{p_u}{4\,e^{3/(2\epsilon')}}$.\vspace{3pt} 

As in our original proof, let $X_u=1$ if $u$ obtains $m$ in the inform phase; otherwise, let $X_u=0$. Then $Pr(X_u=1) \geq 1 - (1 - \frac{p_u}{4\,e^{3/(2\epsilon')}})^{s_A} - O(1/n^{c'})$  where the last term is the probability that $s_N$ or $s_A$ deviate by more than $\delta$ from their respective expected values and $c'>0$ is some constant.  Redefine $p_u = \frac{16\,e^{3/(2\epsilon')}}{\epsilon'(1-\delta')2^i}$; this is a constant factor increase, so the cost to each node is asymptotically equal. Then, $Pr(X_u=1)
\geq 1-(1-\frac{4}{\epsilon'\,(1-\delta')\,2^i})^{(1-\delta')\,2^{i/2}\ln n} - O(1/n^{c'})$ $\geq 1 - e^{-4\ln n/(\epsilon'\,2^{i/2})} - O(1/n^{c''}) \geq \frac{2\ln n}{\epsilon' 2^{i/2}} - O(1/n^{c''}) \geq \frac{\ln n}{\epsilon' 2^{i/2}}$ for sufficiently large $n$. We can then apply Theorem~\ref{theorem:concentration} as in the original proof and obtain the desired result.\vspace{\myval} 
\end{proof}

\noindent The proofs for Lemmas~\ref{lemma:inform2} and~\ref{lemma:inform_phase2} can be redone in a similar fashion. Now we show that communication occurs in the final round. The modifications to the sending and listening probabilities, and sending of the decoy messages, increases costs by a constant factor. \vspace{\myval}

\begin{lemma}\label{lemma:reactive_success}
With high probability, the modified version of \AlgL~guarantees that at least $(1-\epsilon')n$   nodes become informed when $f<1/24$ and Carol is reactive.\vspace{\myval} 
\end{lemma}
\begin{proof}
Again, call a slot {\it active} if it contains either $m$ or noise. Using the new definition of a blocked phase defined in the proof of Lemma~\ref{lemma:inform1} above, note that Carol and her Byzantine nodes cannot block a phase containing at least $4Cfn^{3/2}$ active slots. We can make the same argument as in Lemma~\ref{lemma:competitive_costs} by using $\beta$ rather than $1/4$ as the fraction of jammed slots that constitute a blocked phase; however, for simplicity we stick with $1/4$ noting that this does not affect correctness. From the proof above, w.h.p., at least $(1-\delta)(1/2)2^{(3/2)i}$ slots in a phase are active. For concreteness, set $\delta=1/2$ which implies that, w.h.p, at least $(1/4)\,2^{(3/2)i}$ slots are active. Then, solving for $i$ in $2^{(3/2)i}/4 = 4C\,f\,n^{3/2}$ tells us that, w.h.p, Carol and her Byzantine nodes cannot block round $i = \frac{2}{3}\lg(16\,C\,f) + \lg n$. By our new Lemma~\ref{lemma:inform1}, and by modifying Lemmas~\ref{lemma:inform2} and~\ref{lemma:inform_phase2}, at least $(1-\epsilon')n$ correct nodes will become informed in this round. 

As shown in Lemma~\ref{lemma:competitive_costs}, we must ensure that Alice and the correct nodes do not exceed their respective budgets.  When executing \AlgL, there exists some constant $d'>0$ such that the  cost to each node in round $i$ is at most $d'\,2^{i/2}$. The cost to a correct node $u$ is at most $d'2^{i/2} + ( \frac{3}{4n} + p_u)\,2^{(3/2)i}$ where $p_u = \frac{16\,e^{3/(2\epsilon')}}{\epsilon'(1-\delta')2^i}$ and we can set $\delta'=1/2$. Substituting for $i$, the cost to $u$ for this round is at most $\mathcal{C}_u = (d'(16Cf)^{1/3} + 12\,C\,f + d''\,(C\,f)^{1/3})n^{1/2}$ where $d''>0$ is some constant depending only on $\epsilon$. Because rounds double in size, the total cost to $u$ up to and including this round is at most $2\,\mathcal{C}_u$. Solving for $C$ in $ C\,n^{1/2} \geq 2\,\mathcal{C}_u$ yields $C \geq (\frac{2d'(16f)^{1/3} + 2d''f^{1/3}}{1-24\,f})^{3/2}$. Therefore, for $f<1/24$ and sufficiently large $C$, the correct nodes do not exceed their budget.\vspace{\myval} 
\end{proof}

We note that $f<1/24$ is an artifact of the constants used to define a blocked phase and to provide the w.h.p. guarantees; it seems likely that this can be improved. However, the crucial point is that our modified \AlgL~is resource competitive against a reactive Carol who controls $\Theta(n)$ Byzantine nodes, and correct nodes can bear the costs for tolerating such a reactive adversary; we do not rely on an external and free source of noise.
\vspace{0pt}


\subsection{System-Size Parameters}\label{section:size}\vspace{0pt}

As stated, the sending and listening probabilities in our protocol require knowledge of $\ln n$ and $1/n$. However, the guarantees provided by \AlgL~still hold if each node has a constant-factor approximation to these values. Such approximations can be used instead of the true values while incurring only a constant-factor increase in cost. There are well-known ``folklore'' algorithms  for efficiently obtaining such approximations in a distributed setting and we may hope that these are executed prior to a jamming attack. 

If such approximations are not possible, our protocol still functions if all nodes share the same polynomial overestimate of $n$; that is, $\nu_u = n^{c'}$ for any constant $c'\geq 1$.  Each node obtains the constant-factor approximation  $\ell_u = \lceil c\ln n \rceil$ to use in our protocol.  In the propagation phase of a fixed round $i$, each step (see Figure~\ref{fig:pseudocode_general_k}) is executed $g$ times  with informed nodes sending with probability $\frac{1}{2^i\,2^g}$ where $g=1, ..., \ell_u$. At some point, $g= \ln n $ and, therefore, each informed node will complete that step of the phase with the correct sending probability to within a factor of $2$. The same technique can be used in the request phase. In this case, the cost of executing~\AlgL~increases by a logarithmic factor and, consequently, the guarantees hold so long as there exists a large, but sublinear, $O(\frac{n}{\ln n})$ number of Byzantine nodes.\vspace{-3pt}

\section{Conclusion and Future Work}

As the size of WSN devices decrease, communication protocols must satisfy the strict energy constraints that are unavoidable at this scale while remaining robust to malicious attacks. Our results address this challenge by demonstrating the feasibility of a critical communication primitive in the face of a powerful adversary who controls $\Theta(n)$ devices in a dense WSN. Moreover, the correct devices enjoy a significant advantage in terms of energy expenditure. A critical open question is whether these resource-competitive results have an analogue in multi-hop WSNs. It would also be of interest to examine other fundamental distributed communication problems, such as consensus or leader election, from a resource-competitive perspective. \smallskip

\noindent {\bf Acknowledgements:} We thank Jared Saia and Valerie King for their invaluable comments and suggestions. 


\bibliographystyle{abbrv}
\bibliography{jam}
\end{document}